\documentclass[11pt]{article}

\usepackage[margin=1in]{geometry}
\usepackage[T1]{fontenc}
\usepackage{lmodern}
\usepackage{microtype}
\usepackage{amsmath,amssymb,amsthm,mathtools}
\usepackage[hidelinks,hypertexnames=false]{hyperref}
\hypersetup{
pdftitle={Optimal Deterministic Fully Sparse Matrix Multiplication},
pdfauthor={Omar Graia}
}
\usepackage{booktabs}
\usepackage{tabularx}
\usepackage{array}
\usepackage{float}
\usepackage{graphicx}
\usepackage[ruled]{algorithm2e}
\usepackage{ifpdf}
\newcolumntype{S}{>{\raggedright\arraybackslash}X}
\ifpdf
  \newcolumntype{T}{>{\raggedright\arraybackslash}p{0.3\textwidth}}
\else
  \newcolumntype{T}{>{\raggedright\arraybackslash}X}
\fi
\newtheorem{theorem}{Theorem}[section]
\newtheorem{lemma}[theorem]{Lemma}
\newtheorem{proposition}[theorem]{Proposition}
\newtheorem{corollary}[theorem]{Corollary}
\theoremstyle{definition}
\newtheorem{definition}[theorem]{Definition}
\theoremstyle{remark}

\newtheorem{claim}{Claim}

\DeclareMathOperator{\nnz}{nnz}
\DeclareMathOperator{\supp}{supp}
\newcommand{\din}{\delta_{\mathrm{in}}}
\newcommand{\dout}{\delta_{\mathrm{out}}}
\newcommand{\eps}{\varepsilon}
\newcommand{\ceil}[1]{\left\lceil #1\right\rceil}

\newcommand{\polylog}{\operatorname{polylog}}
\newcommand{\SafeRec}{\operatorname{SafeRec}}

\title{Optimal Deterministic Fully Sparse Matrix Multiplication}
\author{Omar Graia}
\date{August 19, 2026}

\begin{document}
\maketitle

\begin{abstract}
We give the first deterministic algorithm for fully sparse matrix multiplication that attains the optimal running-time exponent. This result matches the best previously known randomized algorithm running-time exponent. Given compatible matrices $A$ and $B$ over an arbitrary associative ring with identity, with $\operatorname{nnz}(A),\operatorname{nnz}(B)=O(n^{\delta_{\mathrm{in}}})$ and $\operatorname{nnz}(AB)=O(n^{\delta_{\mathrm{out}}})$, our algorithm finds the support of $AB$ and computes the product exactly in
$$O\!\left(n^{\beta_R(\delta_{\mathrm{in}},\min\{\delta_{\mathrm{out}},2\delta_{\mathrm{in}}\})+\varepsilon}\right)$$
operations, where $\beta_R(\delta_{\mathrm{in}},\delta)$ denotes the maximum of $\delta_{\mathrm{in}}$ and $\omega_{\delta_{\mathrm{in}},R}(a,1,b)$ over all $a,b\in[0,1]$ satisfying $a+b=\delta$.
For dense inputs over a commutative ring, this bound simplifies to $O(n^{\omega_R((\delta_{\mathrm{out}}-1)_+,1,1)+\varepsilon})$. With the current rectangular matrix multiplication bounds, this is nearly quadratic, namely $O(n^{2+\varepsilon})$, for every $\delta_\mathrm{out}\le1.321334$, improving the previous deterministic range of $\delta_{\mathrm{out}}\le 0.642668$. To prove this result, we develop a general deterministic recovery technique that finds and fixes sparse parts of an unknown matrix while keeping temporary errors in denser parts under control.
\end{abstract}

\section{Introduction}
Matrix multiplication is a central algorithmic problem in computer science. Its dense complexity is bounded by the matrix multiplication exponent $\omega$, whose current upper bound was recently shown by Google DeepMind to be $\omega<2.371177$~\cite{DupontEtAl}, improving the previous bound $\omega<2.371339$ of Alman et al.~\cite{AlmanDuanEtAl}. A natural question asks whether matrix products can be computed faster when the input matrices, the output matrix, or both, are sparse i.e., having a large number of entries being $0$. This question has been heavily studied from algebraic, combinatorial, database, graph-algorithmic, communication, and numerical perspectives~\cite{Gustavson,YusterZwick,AmossenPagh,IwenSpencer,PaghCompressed,JacobStockel,VanGuchtEtAl,AbboudEtAl,GaoSurvey}.

In fully sparse matrix multiplication, the input is two matrices $A$ and $B$ whose dimensions are at most $n$ that are represented by canonical sparse lists. For sparsity exponents $\din,\dout\in[0,2]$, the $(\din,\dout)$-FSMM problem satisfies $\nnz(A),\nnz(B)=O(n^{\din})$ and $\nnz(AB)=O(n^{\dout})$. The constants hidden in these sparsity bounds are fixed promise constants and may be used by the algorithm. The goal of this problem is to output the canonical sparse list of $AB$.

Note the support of $AB$ is not supplied to the algorithm which is the main difficulty. Even when we know that the output is sparse, it is not immediately clear as to which rows or columns contain the output entries, and cancellations over a general ring can make the support of the product different from what the input supports alone would suggest. However, random sampling and verification have been shown to effectively reveal the unknown support efficiently which is why the previous strongest bounds were randomized algorithms \cite{AbboudEtAl,BennettEtAl}. Earlier deterministic algorithms either worked over a more restricted domain or were inefficient, paying a larger exponent \cite{Kutzkov,Kunnemann,BennettEtAl}.

This paper gives an optimal deterministic algorithm over arbitrary rings. The zero ring is trivial, while the optimality statement for the running-time exponent concerns nontrivial rings. The algorithm combines two levels of deterministic sparse recovery with one within columns and one across columns. We update a residual matrix which is the difference between the true product and the algorithm’s current approximation.
The inner recovery matrix sketches entries within each residual column. The resulting column sketches are treated as elements of a direct-product ring, and an outer recovery matrix identifies the nonzero sketches. A geometric sequence of thresholds eventually makes every residual column sparse enough for exact inner recovery. A bounded total decoder controls the size of speculative updates on columns that are still too dense. When the unknown matrix is $C=AB$, each required two-sided sketch is evaluated by one sparse rectangular product through $H(AB)G^{\mathsf T}=(HA)(BG^{\mathsf T})$.

\subsection{Main Results}
For $a,b\ge0$, let $\omega_{\din,R}(a,1,b)$ be the deterministic exponent for multiplying an $n^a\times n$ matrix by an $n\times n^b$ matrix over $R$ when each factor contains at most $O(n^{\din})$ nonzeros. We give a more formal definition in Section~\ref{sec:preliminaries}. Define 
$$\beta_R(\din,\delta):=\max\!\left\{\din,\max_{\substack{a,b\in[0,1]\\a+b=\delta}}\omega_{\din,R}(a,1,b)\right\}.$$
The output sparsity constraint is bounded by $\nnz(AB)\le\nnz(A)\nnz(B).$ Hence, when $\delta_\text{out}>2\delta_\text{in}$, $\nnz(AB)=O(n^{\dout})$ gives no extra information than the constraint with exponent $2\delta_\text{in}$. This is why the theorem below evaluates the function $\beta_R$ at $\min\{\dout, 2\din\}$.

\begin{theorem}[Deterministic FSMM upper bound]
\label{thm:main}
Let $R$ be an associative ring with identity, let $\din,\dout\in[0,2]$, and let $\eps>0$. Then $(\din,\dout)$-FSMM over $R$ can be solved deterministically and exactly in
$$O\!\left(n^{\beta_R(\din,\min\{\dout,2\din\})+\eps}\right)$$
ring and word operations.
\end{theorem}

In Section~\ref{sec:optimality}, we adapt the rectangular padding argument underlying Bennett, Gajulapalli, Golovnev, and Warton~\cite[Proposition~4.3]{BennettEtAl} to our parameter range and combine it with Theorem~\ref{thm:main} to show that this exponent is optimal.

For dense inputs, the maximization making up $\beta_R$ simplifies to get the following result.

\begin{corollary}
\label{cor:dense}
Let $R$ be a commutative ring in a setting where the standard bilinear rectangular exponents are defined, let $\dout\in[0,2]$, and let $\eps>0$. Dense-input output-sparse matrix multiplication has a deterministic exact algorithm with running time
$$
O\!\left(n^{\omega_R((\dout-1)_+,1,1)+\eps}\right).
$$
\end{corollary}

The dual rectangular exponent is $\omega_R^\perp:=\sup\{\alpha\ge0:\omega_R(\alpha,1,1)=2\}$. The current bounds $\omega<2.371177$ and $\omega_R^\perp\ge0.321334$ follow from Dupont et al. and Vassilevska Williams, Xu, Xu, and Zhou~\cite{DupontEtAl,VXXZ}.

\begin{corollary}
\label{cor:quadratic}
In every commutative-ring setting in which $\omega_R^\perp\ge0.321334$, dense-input output-sparse matrix multiplication can be solved deterministically in $O(n^{2+\eps})$ operations for every $\dout\le1.321334$.
\end{corollary}

The lower-bound reduction was known previously. The new contribution of this paper is the deterministic arbitrary-ring upper bound that attains it, hence giving an optimal algorithm over every nontrivial ring.

\subsection{Previous Work}
Sparse matrix multiplication has a long history. Gustavson's classical algorithm organizes a sparse product through sparse outer products~\cite{Gustavson}. Yuster and Zwick gave faster input-sparse algorithms by separating high-degree and low-degree coordinates and applying rectangular matrix multiplication to the remaining dense part~\cite{YusterZwick}. Related rectangular techniques appear in work of Kaplan, Sharir, and Verbin~\cite{KaplanSharirVerbin}. Boolean matrix multiplication with sparse outputs and related database problems were studied by Lingas, Amossen and Pagh, Van Gucht et al., and Deep, Hu, and Koutris~\cite{Lingas,AmossenPagh,VanGuchtEtAl,DeepHuKoutris}. Pagh introduced compressed matrix multiplication, and later work gave algorithms whose running time depends on the output sparsity, as well as external-memory bounds~\cite{PaghCompressed,JacobStockel,PaghStockel}.

Compressed sensing can be used for exact recovery when each output column is sparse. Iwen and Spencer observed that one can compute $HAB=(HA)B$ and then recover the columns independently~\cite{IwenSpencer}. This builds on literature on sparse recovery and expander-based measurement matrices~\cite{CandesRombergTao,IndykExplicit,BerindeEtAl,GilbertIndyk,GuruswamiUmansVadhan}. A bound on the total output sparsity is weaker than a separate sparsity bound on every column. Bennett et al. handled this with a deterministic two-pass compressed-sensing algorithm. They first recover the sparse columns and then apply a transposed pass to the residual~\cite{BennettEtAl}.

The main earlier deterministic algorithms for sparse-output matrix multiplication are due to Kutzkov and K\"unnemann. Kutzkov used deterministic group-testing ideas over the reals and obtained $O(n^2+n m_{\mathrm{out}}^2)$ time~\cite{Kutzkov}. K\"unnemann connected output-sparse multiplication with deterministic correction of matrix products over the integers and obtained $\widetilde O(n^2\sqrt{m_{\mathrm{out}}}+m_{\mathrm{out}}^2)$ time~\cite{Kunnemann}. Here $m_{\mathrm{out}}:=\nnz(AB)$. Matrix product verification and correction were also studied by Freivalds, Kimbrel and Sinha, Gasieniec et al., Roche, and Bennett et al.~\cite{Freivalds,KimbrelSinha,GasieniecEtAl,Roche,BennettMMV}. In particular, Bennett et al. studied matrix multiplication verification under a sparsity promise on the error matrix $AB-C$~\cite{BennettMMV}, a closely related sparse-residual setting.

\begin{table}[H]
\centering
\caption{Comparison of fully sparse matrix multiplication algorithms in the nonvacuous range $\dout\le2\din$. When $\dout>2\din$, the same bounds hold after replacing $\dout$ by $2\din$. The support of the product is not supplied to any of the algorithms.}
\label{tab:fully-sparse-comparison}
\small
\setlength{\tabcolsep}{4pt}
\renewcommand{\arraystretch}{1.2}
\begin{tabularx}{\textwidth}{c l c T S}
\toprule
Year & Result & Randomness & Running time & Scope / comment \\
\midrule
2013 & Kutzkov & Deterministic & $O\!\left(n^2+n\,m_{\mathrm{out}}^2\right)$ & Real-valued; ignores input sparsity \\
2018 & K\"unnemann & Deterministic & $\widetilde O\!\left(n^2\sqrt{m_{\mathrm{out}}}+m_{\mathrm{out}}^2\right)$ & Integer-valued; ignores input sparsity \\
2024 & Abboud et al. & Randomized & $O\!\left(n^{\beta_R(\din,\dout)+\eps}\right)$ & General rings; Boolean deterministic \\
2025 & Bennett et al. & Deterministic & $O\!\left(n^{\max\left\{\substack{1+\dout/2,\\[-2pt]\omega_{\din,R}(\dout/2,1,1)}\right\}+\eps}\right)$ & Suboptimal in some regimes \\
2025 & Bennett et al. & Randomized & $O\!\left(n^{\beta_R(\din,\dout)+\eps}\right)$ & Optimal randomized curve \\
\midrule
\textbf{2026} & \textbf{This work} & \textbf{Deterministic} & $\boldsymbol{O\!\left(n^{\beta_R(\din,\dout)+\eps}\right)}$ & \textbf{Optimal in all regimes} \\
\bottomrule
\end{tabularx}
\end{table}

The recent study of fully sparse matrix multiplication was initiated by Abboud, Bringmann, Fischer, and K\"unnemann. They introduced an algorithm that reduces sparse multiplication to smaller rectangular products. They were able to achieve the optimal curve over arbitrary rings but used randomization in their algorithm. They also give deterministic results for Boolean and nonnegative products and connect fully sparse multiplication with all-edge triangle problems~\cite{AbboudEtAl}. Bennett et al. later gave an arbitrary-ring deterministic algorithm with exponent $\max\{1+\dout/2,\omega_{\din,R}(\dout/2,1,1)\}$ and a different randomized algorithm that achieves the optimal curve~\cite{BennettEtAl}.

For dense inputs, their deterministic compressed-sensing bound is nearly quadratic only for $\dout\le0.642668$, while the randomized optimal curve remains nearly quadratic through $\dout\le1.321334$ under the current rectangular bounds. Figure~\ref{fig:osmm-curves} compares dense-input output-sparse exponents for different algorithms. The optimal curve was previously attained only by randomized arbitrary-ring algorithms. The deterministic bound in this paper follows the same blue curve.

\begin{figure}[H]
\centering
\includegraphics[width=0.8\textwidth]{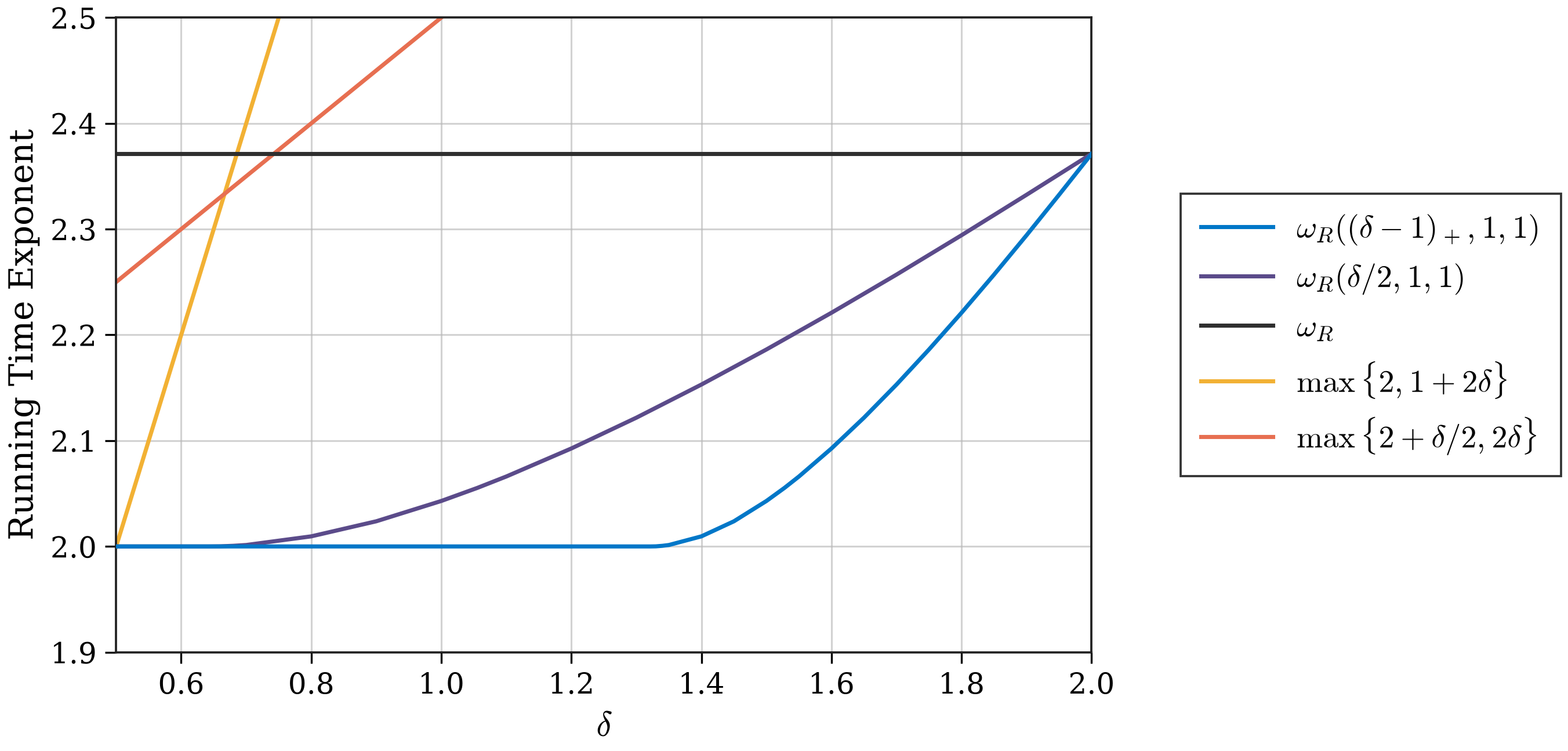}
\caption{Running-time exponents for dense-input output-sparse matrix multiplication, where $\delta=\dout$. The blue curve plots the current upper bound on the optimal exponent $\omega_R((\dout-1)_+,1,1)$, attained randomly in prior work and deterministically here. The purple curve is the deterministic bound of Bennett et al. The orange and red curves are the bounds of Kutzkov and K\"unnemann. The horizontal line is the current upper bound on the square matrix multiplication exponent.}
\label{fig:osmm-curves}
\end{figure}

Their randomized algorithm also uses a geometric sequence of column-sparsity thresholds, but relies on randomized column verification. Our outer deterministic recovery step replaces that verification and lets the inner and outer dimensions vary along the full rectangular curve. Bringmann, Fischer, and Nakos recently studied a robust version in which the goal is to recover a sparse approximation with error controlled by the output tail~\cite{BringmannFischerNakos}.

\subsection{Technical Overview}
Here we give an overview of our algorithm. Let $C=AB$, let $K$ be an upper bound on $\nnz(C)$, and let $D$ be the current approximation for $C$ that the algorithm uses. The algorithm uses rounds with scales $t=1, 2,4,\dots$ and at scale $t$, its goal is to make every column of $C$ with at most $t$ nonzero entries corrected in $D$ exactly.

The difference $E:=C-D$ is the residual matrix. Its $j$th column is the correction that must be added to the current column of $D$, since
$$D_{*j}+E_{*j}=D_{*j}+(C-D)_{*j}=C_{*j}.$$
Thus the algorithm recovers residual columns and adds them to $D$. If $E_{*j}$ is recovered exactly, then the update $D_{*j}\leftarrow D_{*j}+E_{*j}$ makes column $j$ equal to $C_{*j}$.

Now suppose that every column of $C$ with at most $t/2$ nonzeros was corrected in the previous round. Any column that is still incorrect must then have more than $t/2$ nonzeros in $C$. Since $C$ has at most $K$ nonzeros in total, only $O(K/t)$ columns can remain incorrect. This allows the outer sparse recovery step to find the remaining incorrect columns.

We now explain how a column is recovered. We choose an inner measurement matrix $H$ that only recovers vectors with sparsity of $t$ plus the maximum support already stored in a column $D$. We need to account for extra support stored in the columns since a denser column of $D$ may contain extra entries from earlier rounds of recovery targeting sparser columns. For every column $j$, define the measurement $\xi_j:=H(C-D)_{*j}=HE_{*j}.$ If column $j$ of $C$ has at most $t$ nonzeros, then the residual column $E_{*j}=C_{*j}-D_{*j}$ has support at most the support of $C_{*j}$ plus the support already stored in $D_{*j}$. This is within the inner recovery bound, so $\xi_j$ determines $E_{*j}$ exactly. Adding the recovered residual to $D_{*j}$ then makes the column equal to $C_{*j}$.

The algorithm still needs to find which measurements $\xi_j$ are nonzero. Each $\xi_j$ is a vector in $R^m$, where $m$ is the number of rows of $H$ and we view each such vector as one element of the direct-product ring $R^m$. The vector $\xi=(\xi_1,\ldots,\xi_c)\in(R^m)^c$ has nonzero coordinates only when the corresponding column of $D$ is incorrect and does not match $C$. Since we know only $O(K/t)$ columns can remain incorrect at this round, $\xi$ has only $O(K/t)$ nonzero coordinates. Hence we can use another measurement matrix $G$ over $R^m$ to recover all of its nonzero coordinates.

The two levels of measurements can be written as
$$(C-D)\xrightarrow{\text{inner measurements }H}\text{column measurements in }R^m\xrightarrow{\text{outer measurements }G}H(C-D)G^{\mathsf T}.$$
The first step takes a measurement of every residual column and the second step finds the nonzero measurements among those columns. Formally, after identifying $(R^m)^g$ with $m\times g$ matrices over $R$, the outer measurements of $\xi$ are exactly $H(C-D)G^{\mathsf T}$.

For every nonzero measurement returned by the outer decoder, the algorithm runs the inner decoder and adds its output to the corresponding column of $D$. If the column has at most $t$ nonzeros then the inner decoder is able to recover the exact residual and hence correct the column. If the column has more than $t$ nonzeros however, then exact recovery is not guaranteed in that round for the column. Even then, any temporary update adds only a bounded number of entries to $D$, and the next round includes those stored entries in its recovery bound. Thus, columns with more nonzero entries will still be fully corrected in later rounds. 

At each round, the scale doubles. At scale $t$ the inner recovery sparsity is $\widetilde O(t)$ and the outer recovery sparsity is $O(K/t)$ so their product is $\widetilde O(K)$. This keeps the total decoding and update work within $K^{1+\alpha}\polylog n$ for every fixed $\alpha>0$. At a large enough scale, every column would be within the inner recovery bound and at that point all columns of $D$ would equal the corresponding columns of $C$. 

The algorithm is able to find the measurements of $E=C-D$ without actually computing $C$ by noting that $H(C-D)G^{\mathsf T}=(HA)(BG^{\mathsf T})-HDG^{\mathsf T}.$ Since $H$ and $G$ are binary with polylogarithmic column sparsity, forming $HA$ and $BG^{\mathsf T}$ increases the number of nonzeros by only polylogarithmic factors. As $t$ doubles, the resulting rectangular products cover the dimensions used in $\beta_R(\din,\min\{\dout,2\din\})$.

The rest of the paper proves this overview. Section~\ref{sec:preliminaries} gives the recovery and rectangular-multiplication tools that we use throughout. Section~\ref{sec:nested} proves the general nested recovery theorem. Section~\ref{sec:main-proof} proves Theorem~\ref{thm:main}. Section~\ref{sec:optimality} combines the upper bound with the known rectangular padding reduction to show that this run time is optimal.

\section{Preliminaries}
\label{sec:preliminaries}

\subsection{Notation and computational model.}
For a positive integer $N$, define $[N]=\{1,\ldots,N\}$. For a real number $x$, define $(x)_+:=\max\{x,0\}$. For a vector or matrix $X$ let $\nnz(X)$ be the number of nonzero entries in $X$ and let $\supp(X)$ be their index set. Note that $X_{*j}$ refers to the $j$th column of $X$ and we use $\widetilde O(\cdot)$ to hide factors polynomial in $\log n$.

Throughout this paper, $R$ is an arbitrary associative ring with identity and we do not assume commutativity unless stated otherwise. We count each ring operation, zero test, and operation on an $O(\log n)$-bit index as one unit of cost. Sparse vectors are stored as sorted lists of their nonzero entries, together with each entry's position. Sparse matrices are stored as lists of their sparse columns. Further, merging sparse lists of total length $M$ takes $M\polylog n$ ring and word operations. Computing a matrix product means outputting the sparse list of all its nonzero entries, hence the running time includes reading the input and writing the output.

For the run time analysis, we use the following sparse rectangular matrix-multiplication exponent, defined in the same computational model as FSMM.

\begin{definition}[Sparse rectangular exponent]
\label{def:sparse-exponent}
Let $a,b,c,\delta\ge0$. The exponent $\omega_{\delta,R}(a,b,c)$ is the infimum over all exponents $\gamma$ such that, for every $\eta>0$, there is a deterministic algorithm that multiplies an $\ceil{n^a}\times\ceil{n^b}$ matrix by an $\ceil{n^b}\times\ceil{n^c}$ matrix over $R$ in $O(n^{\gamma+\eta})$ ring and word operations, assuming each input matrix has at most $O(n^\delta)$ nonzero entries and is given as a canonical sparse list. The algorithm must output the requested product as a canonical sparse list.
\end{definition}

In particular, for every $\eta>0$ there is such an algorithm running in $O(n^{\omega_{\delta,R}(a,b,c)+\eta})$ ring and word operations.

The sparse rectangular exponent uses the same computational model as FSMM, including the cost of writing the sparse output~\cite{AbboudEtAl,BennettEtAl}. If $O(n^\delta)$ exceeds the number of entries in an input matrix, then the sparsity bound gives no extra restriction. We use $\omega_{\delta,R}(a,1,b)$ when the middle dimension is $n$. For dense inputs, $\omega_R(a,b,c)$ denotes the corresponding exponent with no sparsity condition. Because the dense exponent already accounts for all output entries, explicitly writing and zero-testing the dense product does not change the exponent. Hence $\omega_{2,R}(a,1,b)=\omega_R(a,1,b)$ for all $a,b\in[0,1]$.

Notice that if $R$ is nontrivial and $\delta\ge\max\{a,c\}$, then we have that
$$\omega_{\delta,R}(a,1,c)\ge a+c.$$
This is because if we place $1_R$ in the same inner coordinate of every row of the left factor and every column of the right factor then the left and right factors will have $\Theta(n^a)$ and $\Theta(n^c)$ nonzero entries respectively, and all $\Theta(n^{a+c})$ entries of the product are equal to $1_R$. Thus, any algorithm following the model of Definition~\ref{def:sparse-exponent} must spend $\Omega(n^{a+c})$ time writing the output.

Now we prove two straightforward consequences of the definition that will be used in the run time analysis. 

\begin{lemma}
\label{lem:padding}
Let $a,b,\delta\ge0$.
\begin{enumerate}
\item If $a'\ge a$ and $b'\ge b$, then
$$
\omega_{\delta,R}(a,1,b)\le\omega_{\delta,R}(a',1,b')
\le\omega_{\delta,R}(a,1,b)+(a'-a)+(b'-b).
$$
\item Replacing the input sparsity bound $O(n^\delta)$ by $O(n^\delta\polylog n)$ does not change the exponent.
\end{enumerate}
\end{lemma}

\begin{proof}
For $(1)$, the first inequality follows by zero padding. Any $n^a\times n$ and $n\times n^b$ instance can be enlarged to dimensions $n^{a'}\times n$ and $n\times n^{b'}$ by adding zero rows and columns. Restricting the padded product to the original rows and columns recovers the original product, so $\omega_{\delta,R}(a,1,b)\le\omega_{\delta,R}(a',1,b')$.

Now for the second inequality, we split the first matrix into $O(n^{a'-a})$ row blocks with each having at most $n^a$ rows, and then split the second matrix into $O(n^{b'-b})$ column blocks with each having at most $n^b$ columns. From there, every row block and column block still has at most $O(n^\delta)$ nonzero entries so we can thus multiply each pair using the algorithm for dimensions $n^a\times n$ and $n\times n^b$. There are $O(n^{(a'-a)+(b'-b)})$ pairs, and putting the resulting blocks together gives
$$\omega_{\delta,R}(a',1,b')\le\omega_{\delta,R}(a,1,b)+(a'-a)+(b'-b).$$

For $(2)$, one inequality is immediate since an $O(n^\delta)$ sparsity bound is stronger than an $O(n^\delta\polylog n)$ bound. For the converse, fix $\gamma>\omega_{\delta,R}(a,1,b)$ and $\zeta>0$. Suppose first that $\delta>0$. Choose $\theta>0$ so that $(1+\theta)(\gamma+\zeta/4)\le\gamma+\zeta$. For all sufficiently large $n$, $\polylog n\le n^{\delta\theta}$. Taking $N=\ceil{n^{1+\theta}}$ and padding to dimensions $N^a\times N$ and $N\times N^b$, each factor has $O(N^\delta)$ nonzeros. Definition~\ref{def:sparse-exponent} therefore gives a running time
$$O\!\left(N^{\gamma+\zeta/4}\right)=O\!\left(n^{\gamma+\zeta}\right).$$
The original product is obtained by restricting the padded product to its original rows and columns.

If $\delta=0$, let $s_A$ and $s_B$ be the numbers of nonzero entries in the two inputs. Then $s_A,s_B=\polylog n$. Considering at most $s_As_B=\polylog^2 n$ pairs, multiply the stored entries directly. Both sorting the resulting entries by position and combining equal positions also take polylogarithmic time each, hence the whole computation takes $\polylog n=O(n^\zeta)$ for every $\zeta>0$, thus the exponent remains $0$.

Since $\gamma>\omega_{\delta,R}(a,1,b)$ and $\zeta>0$ were arbitrary, the extra polylogarithmic factor in the sparsity bound does not change the exponent.
\end{proof}

Lemma~\ref{lem:padding}(1) also implies that $(a,b)\mapsto\omega_{\delta,R}(a,1,b)$ is continuous, since changing the two outer exponents by $\Delta a$ and $\Delta b$ changes the sparse rectangular exponent by at most $|\Delta a|+|\Delta b|$. Hence the maximum in the definition of $\beta_R(\din,\delta)$ is attained.

For dense rectangular exponents, we use the properties developed by Lotti and Romani~\cite{LottiRomani}.

\begin{proposition}[Lotti and Romani~\cite{LottiRomani}]
\label{prop:dense-properties}
Let $R$ be a commutative ring in a setting where the usual bilinear rectangular exponents are defined. The exponent $\omega_R(a,b,c)$ is invariant under permutations of its arguments, is homogeneous, continuous, is nondecreasing in every argument, and satisfies
$$\omega_R(a_1+a_2,b_1+b_2,c_1+c_2)\le\omega_R(a_1,b_1,c_1)+\omega_R(a_2,b_2,c_2).$$
Moreover, $\max\{a+b,a+c,b+c\}\le\omega_R(a,b,c)\le a+b+c$.
\end{proposition}

We will also use the inequality $\omega_R(a,b,c)\le\omega_R(a,1,b+c-1)$ for $a,b,c\in[0,1]$ with $b+c\ge1$, as stated by Bennett et al.~\cite[Corollary~2.6]{BennettEtAl}.

\subsection{Deterministic sparse recovery.}
A recovery pair for vectors in $R^N$ with sparsity $k$ consists of a binary matrix $H\in\{0,1\}^{m\times N}$ and a deterministic decoder $\operatorname{Rec}_{H,k}$. Exact recovery means that $\operatorname{Rec}_{H,k}(Hx)=x$ for every $x\in R^N$ with $\nnz(x)\le k$.

Bennett et al.~\cite[Theorem~3.2]{BennettEtAl}, building on the compressed-sensing construction of Berinde et al.~\cite{BerindeEtAl} and the expander construction of Guruswami, Umans, and Vadhan~\cite{GuruswamiUmansVadhan}, give the following recovery theorem that we use over arbitrary rings.

\begin{theorem}
\label{thm:sparse-recovery}
Let $R$ be a ring and let $\alpha>0$ be a constant. For every universe size $N$ and sparsity $1\le k\le N$, there is a deterministic recovery pair $(H,\operatorname{Rec}_{H,k})$ with
$$\operatorname{rows}(H)=O\!\left(k^{1+\alpha}\polylog N\right).$$
The matrix $H$ is binary and $d$-column-sparse for $d=\polylog N$. A sparse representation of $H$ can be constructed in $\widetilde O(N)$ word operations. The decoder uses $O(k^{1+\alpha}\polylog N)$ operations over $R$, and multiplying $H$ by a sparse vector $x$ takes $\nnz(x)\polylog N$ additions over $R$.
\end{theorem}

The construction of the binary recovery matrix depends only on $N,k$ and $\alpha$, not on a ring. This decoder works over an arbitrary ring and only uses the ring operations stated in the theorem. Hence, we can apply the theorem to rings that vary with $n$, including direct-product rings $R^m$.

If the identity matrix has fewer rows than the recovery matrix, we use it instead. In this case, define $\SafeRec_{I,k}(y):=y$ when $\nnz(y)\le k$, and $\SafeRec_{I,k}(y):=0$ otherwise. Thus, for a constant $C_\alpha$ depending only on $\alpha$, we may assume
$$\operatorname{rows}(H)\le\min\left\{N,C_\alpha k^{1+\alpha}\polylog N\right\}.$$
The identity case occurs only when $N\le C_\alpha k^{1+\alpha}\polylog N$, so reading its input and producing the sparse output takes time within the same bound.

The next lemma extends the decoder to arbitrary inputs, including vectors that may not equal $Hx$ for any $k$-sparse vector $x$.

\begin{lemma}
\label{lem:total-decoder}
For the recovery pair in Theorem~\ref{thm:sparse-recovery}, there is a deterministic algorithm $\SafeRec_{H,k}$ such that, for every input vector $y\in R^m$,
\begin{enumerate}
\item $\SafeRec_{H,k}(y)$ is a canonical sparse vector with at most $k$ nonzero entries;
\item if $x\in R^N$ is $k$-sparse, then $\SafeRec_{H,k}(Hx)=x$;
\item the algorithm uses $O(k^{1+\alpha}\polylog N)$ ring and word operations, whether or not $y=Hx$ for a $k$-sparse vector.
\end{enumerate}
\end{lemma}

\begin{proof}
We use the concrete recovery algorithm in Algorithms~5 and 6 of Bennett et al.~\cite[Appendix~A.3]{BennettEtAl}. For an input $y\in R^m$, that algorithm performs at most $\ceil{\log_2(2k)}$ iterations. In each iteration it calls the deterministic \textsc{Reduce} procedure to obtain a proposed vector. 

The \textsc{Reduce} algorithm runs in $m\polylog N$ time for any input vector. This is because it scans the same measurement blocks each time and adds at most one candidate from each block. Thus there are just $O(m)$ candidates, and the proposed vector can be stored as a sparse list within the same time bound. This does not require the input to be of the form $Hx$ for a $k$-sparse vector $x$. The \textsc{Recovery} algorithm then counts the nonzero entries in the proposed vector before multiplying it by $H$. If there are more than $3k/2$, it returns zero immediately. Otherwise the proposed vector has at most $3k/2$ nonzeros, and since $H$ has polylogarithmic column sparsity, multiplying it by $H$ takes $k\polylog N$ operations. Hence every iteration has bounded cost even when $y$ does not come from a sparse vector, and the run time analysis in the proof of Bennett et al.'s Theorem~3.2 gives $O(k^{1+\alpha}\polylog N)$ ring operations in total.

Run the decoder while retaining the proposed sparse vectors. If the \textsc{Recovery} procedure terminates because a proposed vector has more than $3k/2$ nonzeros, return the zero vector. Otherwise, each proposed vector has at most $3k/2$ nonzeros. Since there are at most $\ceil{\log_2(2k)}$ iterations, the total size of all proposed sparse lists is $O(k\log k)$. Combine repeated coordinates and delete zero entries. If the resulting vector has more than $k$ nonzeros, return the zero vector. Otherwise return it. This takes $O(k\log k\polylog N)$ operations, which is within $O(k^{1+\alpha}\polylog N)$. This proves the first and third claims.

If $x$ is $k$-sparse and $y=Hx$, the original decoder returns $x$ exactly. Since $\nnz(x)\le k$, the final support check does not change it. Hence $\SafeRec_{H,k}(Hx)=x$.
\end{proof}
Exact recovery also implies injectivity. If $x$ is $k$-sparse and $Hx=0$, then
$x=\SafeRec_{H,k}(Hx)=\SafeRec_{H,k}(0)=0.$

\subsection{Direct-product rings and binary matrices.}
For $m\ge1$, let $S=R^m$. Addition and multiplication in $S$ are performed separately in each coordinate, and its identity is $1_S=(1_R,\ldots,1_R)$. Thus one ring operation or zero test in $S$ takes $O(m)$ operations or zero tests in $R$. We view a vector in $S^g$ as an $m\times g$ matrix over $R$, with one element of $S$ in each column.

If $G\in\{0,1\}^{g\times c}$ is a binary matrix and $T$ is a ring, write $G^{(T)}$ for the matrix obtained by interpreting each bit as $0_T$ or $1_T$. For every $x\in T$, we have $0_Tx=x0_T=0_T$ and $1_Tx=x1_T=x$, even when $T$ is noncommutative. Thus $G^{(R)}$ and $G^{(S)}$ have $1$-entries in the same positions, so one sparse list can represent both. Hence, we do not store a separate copy of $1_S$ for each nonzero entry of $G$. We omit the superscript when the ring is clear.

\section{Nested Deterministic Recovery}
\label{sec:nested}
We first study a more general recovery problem that does not depend on matrix multiplication. Let $C\in R^{r\times c}$ be unknown with $r,c\le n$ and $\nnz(C)\le K$ for an integer $K\ge0$. For binary recovery matrices $H$ and $G$ chosen by the algorithm, assume that the algorithm can compute $HC\bigl(G^{(R)}\bigr)^{\mathsf T}.$ In Section~\ref{sec:main-proof} we show how to compute these two-sided measurements when $C=AB.$

We begin by defining the parameters used by the recovery algorithm. First replace $K$ with  $\min\{K,rc\}$. If $K=0$ then we return the zero matrix. Otherwise set $t_{-1}:=0$ and $Q_{-1}:=0$, and for $i\ge0$ define $t_i:=2^i$,
$q_i:=\min\{r,t_i+Q_{i-1}\},$ and $ Q_i:=Q_{i-1}+q_i.$ The value $t_i$ is the true column sparsity that is corrected by the end of round $i$. The inner decoder capacity $q_i$ includes both the true support and the support already stored in a column. The value $Q_i$ bounds the stored support after round $i$. 

Also define $
p_0:=\min\{c,K\} $ and $
p_i:=\min\!\left\{c,\ceil{K/t_{i-1}}\right\}\text {for }i\ge1,$ and let $L:=\ceil{\log_2\min\{r,K\}}$. At round $i$ choose an inner recovery pair $(H_i,\operatorname{Rec}_{H_i,q_i})$ over $R$, where $H_i\in\{0,1\}^{m_i\times r}$. Let $S_i:=R^{m_i}$ and choose an outer recovery pair $(G_i,\operatorname{Rec}^{S_i}_{G_i,p_i})$ over $S_i$, where $G_i\in\{0,1\}^{g_i\times c}$. The same binary matrix $G_i$ is used in both places. In the two-sided measurement, its $0$-$1$ entries are taken in $R$, while the outer decoder uses the same entries in $S_i=R^{m_i}$.

The algorithm starts with  $D^{(0)}=0$. At round $i$ it forms $$
W_i:=H_i(C-D^{(i)})\bigl(G_i^{(R)}\bigr)^{\mathsf T},$$ applies the outer decoder $\SafeRec^{S_i}_{G_i,p_i}$ to recover the vector of measurements $\bigl(H_i(C-D^{(i)})_{*j}\bigr)_{j\in[c]}$, and applies $\SafeRec_{H_i,q_i}$ to each recovered nonzero measurement. The modified columns are then canonicalized.

\begin{algorithm}[H]
\caption{Nested deterministic recovery from two-sided measurements}
\label{alg:abstract-nested}
\DontPrintSemicolon
\SetKwInput{KwInput}{Input}
\SetKwInput{KwOutput}{Output}
\SetKw{KwRet}{return}
\KwInput{Dimensions $r,c$; an upper bound $K\ge\nnz(C)$; access to $HC(G^{(R)})^{\mathsf T}$.}
\KwOutput{The canonical sparse list of $C$.}
$K\leftarrow\min\{K,rc\}$\;
\If{$K=0$}{
\KwRet the empty list\;
}
$D\leftarrow0$; $Q\leftarrow0$; $T\leftarrow\min\{r,K\}$; $L\leftarrow\ceil{\log_2T}$\;
\For{$i\leftarrow0$ \KwTo $L$}{
$t\leftarrow2^i$; $q\leftarrow\min\{r,t+Q\}$\;
\eIf{$i=0$}{
$p\leftarrow\min\{c,K\}$\;
}{
$p\leftarrow\min\{c,\ceil{K/2^{i-1}}\}$\;
}
Construct an inner recovery pair $(H,\operatorname{Rec}_{H,q})$ over $R$\;
$S\leftarrow R^{\operatorname{rows}(H)}$\;
Construct an outer recovery pair $(G,\operatorname{Rec}^{S}_{G,p})$ over $S$\;
$W\leftarrow HC(G^{(R)})^{\mathsf T}-HD(G^{(R)})^{\mathsf T}$\;
$\xi\leftarrow\SafeRec^{S}_{G,p}(W)$, viewing the columns of $W$ as elements of $S$\;
\ForEach{$j\in\supp(\xi)$}{
$z\leftarrow\SafeRec_{H,q}(\xi_j)$\;
$D_{*j}\leftarrow D_{*j}+z$\;
}
Canonicalize the modified sparse columns of $D$\;
$Q\leftarrow Q+q$\;
}
\KwRet $D$\;
\end{algorithm}

\medskip
The first lemma explains how the two recovery layers work together.

\begin{lemma}
\label{lem:column-measurements}
Let $E\in R^{r\times c}$ have at most $p$ nonzero columns. Let $H\in\{0,1\}^{m\times r}$ and define $\xi_j:=HE_{*j}\in R^m$ for $j\in[c]$. View $\xi=(\xi_1,\ldots,\xi_c)$ as a vector in $S^c$ over $S:=R^m$. Then $\xi$ is $p$-sparse. For every binary $G\in\{0,1\}^{g\times c}$, the vector $G^{(S)}\xi\in S^g$, represented as an $m\times g$ matrix over $R$, is exactly
$$
HE\bigl(G^{(R)}\bigr)^{\mathsf T}.
$$
If $G$ belongs to an exact $p$-sparse recovery pair over $S$, the outer decoder therefore recovers every measurement $HE_{*j}$ exactly.
\end{lemma}

\begin{proof}
If the $j$th column of $E$ is zero, then $\xi_j=HE_{*j}=0$. Thus $\supp(\xi)$ is contained in the set of nonzero columns of $E$, and $\nnz(\xi)\le p$.

Fix $\ell\in[g]$. The $\ell$th coordinate of $G^{(S)}\xi$ is
$$\bigl(G^{(S)}\xi\bigr)_\ell=\sum_{j=1}^c \bigl(G_{\ell j}1_S\bigr)\xi_j.$$
Since each $G_{\ell j}$ is either $0$ or $1$, the factor $G_{\ell j}1_S$ either makes $\xi_j$ zero or leaves it unchanged. Moreover, multiplication by $G_{\ell j}1_R$ has the same effect on either side. Thus this coordinate is the vector $\sum_{j=1}^c G_{\ell j}HE_{*j}=\sum_{j=1}^c HE_{*j}G_{\ell j}\in R^m$. That vector is the $\ell$th column of $HE(G^{(R)})^{\mathsf T}$. So the displayed matrix represents $G^{(S)}\xi$. Since $\xi$ is $p$-sparse, the bounded outer decoder is exact on this input and returns it in canonical form.
\end{proof}

\begin{lemma}
\label{lem:invariant}
For every $i\in\{0,\ldots,L\}$, right before round $i$ the following statements hold.
\begin{enumerate}
\item Every column $C_{*j}$ with $\nnz(C_{*j})\le t_{i-1}$ equals $D^{(i)}_{*j}$.
\item Every column of $D^{(i)}$ has at most $Q_{i-1}$ nonzero entries.
\end{enumerate}
Immediately after round $i$, the same statements hold with $t_i$, $Q_i$, and $D^{(i+1)}$ in place of $t_{i-1}$, $Q_{i-1}$, and $D^{(i)}$. In particular, $D^{(L+1)}=C$.
\end{lemma}

\begin{proof}
For $i=0$, every column of $C$ with at most $t_{-1}=0$ nonzeros is zero and agrees with $D^{(0)}=0$. Every column of $D^{(0)}$ has support $Q_{-1}=0$.

Assume the two statements hold immediately before round $i$ and let $E^{(i)}:=C-D^{(i)}$. If $i=0$, every nonzero column of $E^{(0)}=C$ contains at least one nonzero entry, so there are at most $\min\{c,K\}=p_0$ such columns. If $i\ge1$, every incorrect column has more than $t_{i-1}$ true nonzeros by the first inductive hypothesis. Since $C$ has at most $K$ nonzeros, the number of incorrect columns is at most $K/t_{i-1}$ and so at most $p_i$. Therefore $E^{(i)}$ has at most $p_i$ nonzero columns.

Apply Lemma~\ref{lem:column-measurements} with $E=E^{(i)}$, $H=H_i$, and $G=G_i$. The bounded outer decoder recovers
$$
\xi^{(i)}=\bigl(H_iE^{(i)}_{*1},\ldots,H_iE^{(i)}_{*c}\bigr)
$$
exactly over $S_i=R^{m_i}$.

Fix a column $j$ with $\nnz(C_{*j})\le t_i$. The second inductive hypothesis gives
$$
\nnz(E^{(i)}_{*j})
\le\nnz(C_{*j})+\nnz(D^{(i)}_{*j})
\le t_i+Q_{i-1}.$$
Clearly since $E^{(i)}_{*j}\in R^r$, it has at most $r$ nonzero entries. Because $\nnz(E^{(i)}_{*j})\le t_i+Q_{i-1}$ we get
$\nnz(E^{(i)}_{*j})\le\min\{r,t_i+Q_{i-1}\}=q_i$. If the residual is zero then the column is already correct. Otherwise, exact recovery of $H_i$ implies $H_iE^{(i)}_{*j}\ne0$, so the outer output contains coordinate $j$. Lemma~\ref{lem:total-decoder} then gives
$$\SafeRec_{H_i,q_i}\bigl(H_iE^{(i)}_{*j}\bigr)=E^{(i)}_{*j}.$$
Adding the update makes column $j$ equal to $C_{*j}$.

For a column with more than $t_i$ true nonzeros, the inner input may have more than $q_i$ nonzeros. Lemma~\ref{lem:total-decoder} still guarantees an update with at most $q_i$ nonzeros. After canonicalization, every current column has support at most $Q_{i-1}+q_i=Q_i$. This proves the statement after round $i$.

Finally, $t_L\ge\min\{r,K\}$. Every column of $C$ has at most $r$ positions and at most $K$ nonzero entries, so every column has support at most $t_L$. The first statement after round $L$ gives $D^{(L+1)}=C$.
\end{proof}

\begin{theorem}
\label{thm:nested}
Let $C\in R^{r\times c}$ satisfy $r,c\le n$ and $\nnz(C)\le K$ for an integer $K\ge0$, and let $\alpha>0$. Suppose that the algorithm can compute $H_iC(G_i^{(R)})^{\mathsf T}$ for every pair $H_i,G_i$ needed by Algorithm~\ref{alg:abstract-nested}. Then Algorithm~\ref{alg:abstract-nested} recovers $C$ exactly and deterministically from $O(1+\log(1+\min\{r,K\}))$ two-sided measurements. In addition to the cost of computing those measurements, it uses
$$\widetilde O\!\left(r+c+K^{1+\alpha}\right)$$
ring and word operations.
\end{theorem}

\begin{proof}
After replacing $K$ by $\min\{K,rc\}$, if $K=0$, the algorithm returns the zero matrix without computing a two-sided measurement. Now assume that $K\ge1$ and let $L:=\ceil{\log_2\min\{r,K\}}$. Algorithm~\ref{alg:abstract-nested} computes one two-sided measurement in each round $i\in\{0,\ldots,L\}$, for a total of $L+1=O(1+\log(1+\min\{r,K\}))$ two-sided measurements. Lemma~\ref{lem:invariant} shows that after round $L$ the current matrix equals $C$, so it remains to bound the work performed outside the computation of these measurements.

We first show that the inner and outer sparsity parameters remain balanced.

\begin{claim}
For every round $i$, $p_iq_i=O(K\log n)$.
\end{claim}

\begin{proof}[Proof of claim]
We first bound $q_i$. We show by induction that $Q_i\le(i+1)2^i$ for every $i\ge0$ and $q_i\le(i+2)2^{i-1}$ for every $i\ge1$. At round $0$, we have $q_0=Q_0=1$. Now suppose that $i\ge1$ and $Q_{i-1}\le i2^{i-1}$. Since the cap at $r$ can only decrease $q_i$, we have $q_i\le2^i+Q_{i-1}\le(i+2)2^{i-1}$. It follows that $Q_i=Q_{i-1}+q_i\le(i+1)2^i$. Thus $q_i=O(t_i\log n)$ and $Q_i=O(t_i\log n)$.

For $i=0$, we have $p_0q_0\le K$. Now fix $i\ge1$. Since $i\le L$, we have $t_{i-1}<\min\{r,K\}$, so $K/t_{i-1}\ge1$. Therefore $p_i\le\ceil{K/t_{i-1}}\le2K/t_{i-1}=4K/t_i$. Combining this with $q_i=O(t_i\log n)$ gives $p_iq_i=O(K\log n)$.
\end{proof}

The proof of the claim also gives $q_i=O(t_i\log n)$ and $Q_i=O(t_i\log n)$, which we will use below. The bounded outer decoder returns at most $p_i$ nonzero measurements in round $i$, and each call to $\SafeRec_{H_i,q_i}$ returns at most $q_i$ nonzero entries. Thus the inner decoders propose at most $p_iq_i=O(K\log n)$ entries in one round and $O(K\log^2 n)$ entries over all rounds.

We next bound the cost of the decoder calls.

\begin{claim}
The outer and inner decoder calls use $\widetilde O(K^{1+\alpha})$ ring and word operations in total.
\end{claim}

\begin{proof}[Proof of claim]
Let $m_i$ be the number of rows of $H_i$. Theorem~\ref{thm:sparse-recovery}, together with the identity fallback, gives $m_i=O(q_i^{1+\alpha}\polylog n)$. Lemma~\ref{lem:total-decoder} shows that the bounded outer decoder performs $O(p_i^{1+\alpha}\polylog n)$ operations over $S_i=R^{m_i}$. Since one operation or zero test in $S_i$ can be performed using $O(m_i)$ ring operations or zero tests in $R$, the outer-decoding cost in round $i$ is
$$O\!\left(m_ip_i^{1+\alpha}\polylog n\right)=O\!\left((p_iq_i)^{1+\alpha}\polylog n\right)=O\!\left(K^{1+\alpha}\polylog n\right).$$
There are at most $p_i$ inner decoder calls in that round. Lemma~\ref{lem:total-decoder} bounds their combined cost by
$$p_i\,O\!\left(q_i^{1+\alpha}\polylog n\right)\le O\!\left((p_iq_i)^{1+\alpha}\polylog n\right)=O\!\left(K^{1+\alpha}\polylog n\right).$$
Summing over the $O(\log n)$ rounds changes only the polylogarithmic factor.
\end{proof}

It remains to account for the recovery matrices, the residual measurement terms, and the sparse updates. The recovery matrices are binary and have polylogarithmic column sparsity. They can be generated in $\widetilde O(r+c)$ time in each round from the construction in Theorem~\ref{thm:sparse-recovery}. Since there are $O(\log n)$ rounds, their total construction time is $\widetilde O(r+c)$. Each outer matrix is stored as a sparse list of its $1$-entries. The matrix $G_i$ is stored as the list of positions where its entries are $1$. We use this same list over $R$ and over $S_i$, so we do not need to store a separate copy of $1_{S_i}$ for each nonzero entry.

The size of the two-sided measurements are also accounted for. In round $i$, the matrix $W_i$ has $m_i g_i$ entries. Theorem~\ref{thm:sparse-recovery} and the identity fallback give $m_i=O(q_i^{1+\alpha}\polylog n)$ and $g_i=O(p_i^{1+\alpha}\polylog n)$, so
$$m_i g_i=O\!\left((p_iq_i)^{1+\alpha}\polylog^2 n\right)=\widetilde O\!\left(K^{1+\alpha}\right).$$
Thus, once the two terms defining $W_i$ have been computed, representing $W_i$ and performing their entrywise subtraction costs $\widetilde O(K^{1+\alpha})$ operations in round $i$. Summing over all rounds changes only the polylogarithmic factor.

Every nonzero entry in $D^{(i)}$ comes from an update proposed by an inner decoder in an earlier round. Since the inner decoders propose $O(K\log^2 n)$ entries over all rounds, $\nnz(D^{(i)})=O(K\log^2 n)$. To compute $H_iD^{(i)}(G_i^{(R)})^{\mathsf T}$, we go through the nonzero entries of $D^{(i)}$ and the corresponding nonzero entries in the relevant columns of $H_i$ and $G_i$. Since both matrices have polylogarithmic column sparsity, this takes $\widetilde O(K)$ operations in each round and $\widetilde O(K)$ operations over all rounds.

Finally, in round $i$, at most $p_i$ columns are modified. Before round $i$, each modified column has at most $Q_{i-1}$ stored entries by Lemma~\ref{lem:invariant}, and the new update has at most $q_i$ entries by Lemma~\ref{lem:total-decoder}. Thus processing all modified columns costs $\widetilde O\!\left(p_i(Q_{i-1}+q_i)\right)=\widetilde O(K)$. For $i\ge1$, this follows from $Q_{i-1}+q_i=O(t_i\log n)$ and $p_i=O(K/t_i)$. Round $0$ also costs $\widetilde O(K)$. Summing over all rounds only adds polylogarithmic factors.

Therefore, the recovery matrices, decoder calls, residual measurements, and sparse updates together use
$$\widetilde O\!\left(r+c+K^{1+\alpha}\right)$$
ring and word operations, not including the cost of the two-sided measurements. This proves the theorem.
\end{proof}

\section{Proof of Main Result}
\label{sec:main-proof}
We are now ready to prove our main result which is Theorem~\ref{thm:main}.
To do this, we apply Theorem~\ref{thm:nested} with $C=AB$. The only extra step is to compute a two-sided measurement efficiently.

\begin{proof}[Proof of Theorem~\ref{thm:main}]
If $R$ is the zero ring, then $AB=0$ for every input. Reading the input and returning the empty list costs $O(n^{\din})$. This is within the claimed bound since $\beta_R(\din,\min\{\dout,2\din\})\ge\din$. We therefore assume that $R$ is nontrivial.

We first remove coordinates that do not affect the product. Let $I_{\mathrm{row}}$ be the nonzero rows of $A$, let $I_{\mathrm{col}}$ be the nonzero columns of $B$, and let $I_{\mathrm{mid}}$ contain the indices that are both nonzero columns of $A$ and nonzero rows of $B$. Entries of $AB$ outside $I_{\mathrm{row}}\times I_{\mathrm{col}}$ are zero, and indices outside $I_{\mathrm{mid}}$ make no contribution to the product. Therefore
$$(AB)[I_{\mathrm{row}},I_{\mathrm{col}}]=A[I_{\mathrm{row}},I_{\mathrm{mid}}]B[I_{\mathrm{mid}},I_{\mathrm{col}}].$$
Sorting and merging the input indices gives active matrices $A_0\in R^{r\times p}$ and $B_0\in R^{p\times c}$ in $\widetilde O(\nnz(A)+\nnz(B))$ operations. Restoring the original row and column labels at the end takes time linear in the output size. If $r=0$, $p=0$, or $c=0$, then $A_0B_0=0$, so the algorithm returns the empty list. From now on, assume $r,p,c\ge1$. Since every row in $I_{\mathrm{row}}$ is a nonzero row of $A$, every column in $I_{\mathrm{col}}$ is a nonzero column of $B$, and every index in $I_{\mathrm{mid}}$ is both a nonzero column of $A$ and a nonzero row of $B$, we have
$$r\le\min\{n,\nnz(A)\},\qquad c\le\min\{n,\nnz(B)\},\qquad p\le\min\{n,\nnz(A),\nnz(B)\}.$$
Therefore
$$r,p,c=O\!\left(n^{\min\{1,\din\}}\right).$$

Let $C:=A_0B_0$ and set $\delta:=\min\{\dout,2\din\}$. Choose a constant $\rho>0$ sufficiently small that $8\rho<\eps$, and choose the recovery parameter $\alpha>0$ so that $2\alpha\le\rho$. Let $C_{\mathrm{out}}$ be the fixed constant in the output-sparsity promise, so $\nnz(C)\le C_{\mathrm{out}}n^{\dout}$, and set
$$K:=\min\left\{\ceil{C_{\mathrm{out}}n^{\dout}},rc,\nnz(A_0)\nnz(B_0)\right\}.$$
Every nonzero entry of $C$ comes from at least one pair of nonzero input entries, so $\nnz(C)\le K$. Since we know $K\le C_{\mathrm{out}}n^{\dout}$ and $K\le\nnz(A_0)\nnz(B_0)=O(n^{2\din})$, we have $K=O(n^{\min\{\dout,2\din\}})=O(n^\delta)$.

For sufficiently large $n$, every fixed constant and polylogarithmic factor below can be absorbed into an extra factor of $n^\rho$. Run Algorithm~\ref{alg:abstract-nested} on $C$ and compute each two-sided measurement by
$$HC(G^{(R)})^{\mathsf T}=H(A_0B_0)(G^{(R)})^{\mathsf T}=(HA_0)\bigl(B_0(G^{(R)})^{\mathsf T}\bigr).$$
The first equality follows from $C=A_0B_0$, and the second follows from associativity. Because no ring entries are reordered, the ring does not need to possess commutativity. Theorem~\ref{thm:nested} thus shows that the algorithm recovers $C$ exactly.

Now, we bound the work outside the rectangular products. There is a running time of $\widetilde O(n^{\din})$ for removing inactive coordinates. In round $i$, each nonzero of $A_0$ can contribute through only the nonzero entries in one column of $H_i$. Since each column of $H_i$ has $\polylog n$ nonzeros, $H_iA_0$ has at most $O(\nnz(A_0)\polylog n)$ contributions. The same argument gives at most $O(\nnz(B_0)\polylog n)$ contributions to $B_0(G_i^{(R)})^{\mathsf T}$. Sorting and combining these contributions constructs both compressed factors in
$$\widetilde O\!\left(\nnz(A_0)+\nnz(B_0)\right)=\widetilde O(n^{\din})$$
operations in each round. There are $O(\log n)$ rounds, so all compressed-factor construction costs $\widetilde O(n^{\din})$ in total.

Theorem~\ref{thm:nested} bounds the remaining recovery overhead by
$$\widetilde O\!\left(r+c+K^{1+\alpha}\right)=\widetilde O\!\left(n^{\din}+n^{\delta(1+\alpha)}\right).$$
We also have $\beta_R(\din,\delta)\ge\delta$. Indeed, $\delta\le2\min\{1,\din\}$, so we may choose $a,b\in[0,\min\{1,\din\}]$ with $a+b=\delta$. Then $a,b\le\din$, and the output-size observation following Definition~\ref{def:sparse-exponent} gives $\omega_{\din,R}(a,1,b)\ge a+b=\delta$. The definition of $\beta_R$ gives the desired inequality. Since $\delta\le2$ and $2\alpha\le\rho$,
$$\delta(1+\alpha)=\delta+\delta\alpha\le\beta_R(\din,\delta)+\rho,$$
so all nonmultiplication work is $O(n^{\beta_R(\din,\delta)+3\rho})$ after absorbing polylogarithmic factors.

It remains to bound the rectangular product in a fixed round $i$. Let $m_i$ and $g_i$ be the numbers of rows of the inner and outer recovery matrices. The identity fallback gives $m_i\le r$ and $g_i\le c$. The parameter bounds in the proof of Theorem~\ref{thm:nested} give $q_i=O(t_i\log n)$ and $p_i=O(1+K/t_i)$. Also, $t_i\le2K=O(n^\delta)$. Therefore, using Theorem~\ref{thm:sparse-recovery} and the bounds above, we obtain
$$\begin{aligned}
m_i
&\le\min\left\{r,C_\alpha q_i^{1+\alpha}\polylog n\right\}\\
&\le\min\left\{O\!\left(n^{\min\{1,\din\}}\right),C_\alpha(t_i\log n)^{1+\alpha}\polylog n\right\}\\
&\le n^{\min\{1,\din,\delta,\log_n t_i\}+2\rho},
\end{aligned}$$
and similarly
$$\begin{aligned}
g_i
&\le\min\left\{c,C_\alpha p_i^{1+\alpha}\polylog n\right\}\\
&\le\min\left\{O\!\left(n^{\min\{1,\din\}}\right),C_\alpha(1+K/t_i)^{1+\alpha}\polylog n\right\}\\
&\le n^{\min\{1,\din,(\delta-\log_n t_i)_+\}+2\rho}.
\end{aligned}$$
Here the $2\rho$ slack absorbs the factors coming from the $(1+\alpha)$ powers, the fixed constants, and the polylogarithmic terms. The active inner dimension satisfies $p\le n$ and may be zero padded to $n$.

The compressed factors remain input-sparse. The contribution bounds above give
$$\begin{aligned}
\nnz(H_iA_0)&=O(\nnz(A_0)\polylog n)=O(n^{\din}\polylog n),\\
\nnz\bigl(B_0(G_i^{(R)})^{\mathsf T}\bigr)&=O(\nnz(B_0)\polylog n)=O(n^{\din}\polylog n).
\end{aligned}$$
By Lemma~\ref{lem:padding}, the polylogarithmic increase in input sparsity does not affect the sparse exponent.

We now make the use of Definition~\ref{def:sparse-exponent} uniform over the rounds. Let $\mathcal G:=\{\min\{j\rho,1\}:0\le j\le\ceil{1/\rho}\}$. For each round $i$, choose $\bar a_i,\bar b_i\in\mathcal G$ so that $m_i\le n^{\bar a_i}$, $g_i\le n^{\bar b_i}$, and
$$\min\{1,\din,\delta,\log_n t_i\}\le\bar a_i\le\min\{1,\din,\delta,\log_n t_i\}+3\rho,$$
$$\min\{1,\din,(\delta-\log_n t_i)_+\}\le\bar b_i\le\min\{1,\din,(\delta-\log_n t_i)_+\}+3\rho.$$
These choices are possible since the bounds above are within $2\rho$ of the target exponents and the values in $\mathcal G$ are spaced by at most $\rho$. Since $\rho$ depends on only $\eps$ we know that $\mathcal G$ has constant size, so only a constant number of exponent pairs can occur. By Lemma~\ref{lem:padding}, allowing $O(n^{\din}\polylog n)$ nonzeros in the compressed factors does not change the sparse exponent. For each pair in $\mathcal G^2$, fix a deterministic algorithm with exponent slack $\rho$. Increasing the two outer exponents to $\bar a_i$ and $\bar b_i$ adds at most $3\rho$ to each one. By Lemma~\ref{lem:padding}, this increases the sparse rectangular exponent by at most $6\rho$ in total. The algorithm for the resulting grid point uses another $\rho$ of exponent slack. Thus the round-$i$ product costs
$$O\!\left(n^{\omega_{\din,R}\left(\min\{1,\din,\delta,\log_n t_i\},1,\min\{1,\din,(\delta-\log_n t_i)_+\}\right)+7\rho}\right).$$

The two outer exponents in this expression lie in $[0,1]$ and have sum at most $\delta$. If $\log_n t_i\le\delta$, their sum is at most $\log_n t_i+(\delta-\log_n t_i)=\delta$. If $\log_n t_i>\delta$, the second exponent is zero and the first is at most $\delta$. We can increase the two exponents until their sum is exactly $\delta$. If their current sum is $s$, we need to add $\delta-s$, while we can add up to $2-s$ before either exponent exceeds $1$. Since $\delta\le2$, we have $\delta-s\le2-s$, so this increase is always possible. Lemma~\ref{lem:padding} and the definition of $\beta_R$ then give
$$\omega_{\din,R}\left(\min\{1,\din,\delta,\log_n t_i\},1,\min\{1,\din,(\delta-\log_n t_i)_+\}\right)\le\max_{\substack{a,b\in[0,1]\\a+b=\delta}}\omega_{\din,R}(a,1,b)\le\beta_R(\din,\delta).$$
Therefore every rectangular product costs $O(n^{\beta_R(\din,\delta)+7\rho})$. There are $O(\log n)$ rounds, and the remaining logarithmic factor is at most $n^\rho$ for all sufficiently large $n$. Since $8\rho<\eps$, the total running time is
$$O\!\left(n^{\beta_R(\din,\delta)+\eps}\right)=O\!\left(n^{\beta_R(\din,\min\{\dout,2\din\})+\eps}\right).$$

Finally, restore the original row and column labels. The output remains canonical, which completes the proof.
\end{proof}

Now that we have proved the general deterministic upper bound, we give two consequences. The first is the dense-input case over a commutative ring, where we simplify the maximization in $\beta_R$. The second uses the best current rectangular matrix-multiplication bounds to find the range where the running time is nearly quadratic.

\begin{proof}[Proof of Corollary~\ref{cor:dense}]
Set $\din=2$. Since $\dout\le2$, the minimum in Theorem~\ref{thm:main} equals $\dout$, and $\omega_{2,R}(a,1,b)=\omega_R(a,1,b)$. Hence
$$
\beta_R(2,\dout)=\max\!\left\{2,\max_{a+b=\dout}\omega_R(a,1,b)\right\}.
$$
If $0\le\dout\le1$, Proposition~\ref{prop:dense-properties} gives $\omega_R(a,1,b)\le a+1+b=1+\dout\le2$ for every $a+b=\dout$. Thus $\beta_R(2,\dout)=2=\omega_R(0,1,1)$.

Assume $1\le\dout\le2$. For $a,b\in[0,1]$ with $a+b=\dout$, permutation symmetry and the rectangular interpolation inequality of Lotti and Romani, as recorded in Bennett et al.~\cite[Corollary~2.6]{BennettEtAl}, give
$$
\omega_R(a,1,b)=\omega_R(1,a,b)
\le\omega_R(1,1,\dout-1)=\omega_R(\dout-1,1,1).
$$
Equality is attained at $(a,b)=(\dout-1,1)$. Therefore $\beta_R(2,\dout)=\omega_R((\dout-1)_+,1,1)$, and the result follows from Theorem~\ref{thm:main}.
\end{proof}

\begin{proof}[Proof of Corollary~\ref{cor:quadratic}]
If $\dout\le1.321334$, then $(\dout-1)_+\le0.321334$. By the definition of $\omega_R^\perp$, monotonicity, and continuity from Proposition~\ref{prop:dense-properties}, the current rectangular bound $\omega_R^\perp\ge0.321334$~\cite{VXXZ} gives $\omega_R((\dout-1)_+,1,1)=2$. Corollary~\ref{cor:dense} therefore gives $O(n^{2+\eps})$ time.
\end{proof}

\section{Optimality}
\label{sec:optimality}
The rectangular padding argument underlying our lower bound appears in Bennett et al.~\cite[Proposition~4.3]{BennettEtAl}. We give the argument in our notation, extend it to the full parameter range needed here, and combine it with Theorem~\ref{thm:main}.

\begin{theorem}[Optimality]
\label{thm:optimal}
Let $R$ be a nontrivial associative ring with identity and let $\din,\dout\in[0,2]$. Suppose that, for some $\gamma\ge0$, $(\din,\dout)$-FSMM over $R$ can be solved deterministically in $O(n^{\gamma+\eps})$ ring and word operations for every constant $\eps>0$. Then
$$
\gamma\ge\beta_R(\din,\min\{\dout,2\din\}).
$$
Consequently, the running-time exponent in Theorem~\ref{thm:main} is optimal.
\end{theorem}

\begin{proof}
By the computational model of Section~\ref{sec:preliminaries}, the running time includes reading the input. Since the two sparse input lists may have total length $\Theta(n^{\din})$, we have $\gamma\ge\din$.

Fix $a,b\in[0,1]$ with $a+b=\min\{\dout,2\din\}$. Consider matrices $X\in R^{\ceil{n^a}\times n}$ and $Y\in R^{n\times\ceil{n^b}}$ with at most $O(n^{\din})$ nonzero entries each. Add zero rows to $X$ and zero columns to $Y$ to obtain square matrices $X',Y'\in R^{n\times n}$. This does not change their sparsities, and the nonzero part of $X'Y'$ is exactly $XY$. Also,
$$\nnz(X'Y')\le\ceil{n^a}\ceil{n^b}=O(n^{a+b})=O(n^{\min\{\dout,2\din\}})=O(n^{\dout}).$$
Thus $(X',Y')$ satisfies the $(\din,\dout)$-FSMM sparsity bounds. Running the assumed FSMM algorithm on $X'$ and $Y'$ and removing the padded rows and columns then computes the canonical sparse list of $XY$ within the same asymptotic running time. Definition~\ref{def:sparse-exponent} therefore gives $\gamma\ge\omega_{\din,R}(a,1,b)$.

Since this holds for every $a,b\in[0,1]$ with $a+b=\min\{\dout,2\din\}$,
$$
\gamma\ge
\max\!\left\{\din,
\max_{\substack{a,b\in[0,1]\\a+b=\min\{\dout,2\din\}}}
\omega_{\din,R}(a,1,b)\right\}
=\beta_R(\din,\min\{\dout,2\din\}).
$$
Theorem~\ref{thm:main} gives a deterministic upper bound with the same exponent, which proves optimality.
\end{proof}

The padding reduction and the matching randomized upper bound were known previously~\cite{AbboudEtAl,BennettEtAl}. Theorem~\ref{thm:main} gives the deterministic arbitrary-ring upper bound needed to match the same exponent without randomization.

\section{Conclusion}
We gave a deterministic exact algorithm for fully sparse matrix multiplication whose exponent matches the optimal randomized curve in every input-output sparsity range. Our main contribution is a nested recovery theorem for a sparse matrix that is not given explicitly. From there, an outer bounded decoder over a direct-product ring is able to find the nonzero inner measurements. A geometric sequence of scales balances the two sparsity levels and keeps the temporary updates under control.

Several questions still remain though. The $k^{1+\alpha}\polylog n$ deterministic recovery method is one source of the exponent slack in our upper bound. A near-linear exact recovery scheme over arbitrary rings would remove this source of overhead. It would also be useful to understand whether nested recovery can work for approximately sparse products, building on the robust reductions of Bringmann, Fischer, and Nakos~\cite{BringmannFischerNakos}. The nested recovery theorem may also apply when an unknown sparse object is the output of a bilinear computation whose measurements factor efficiently.

\end{document}